\documentclass[10pt,a4paper]{article}
\usepackage[british]{babel}
\usepackage[final]{microtype}
\usepackage[T1]{fontenc}
\usepackage[utf8]{inputenc}
\usepackage{lmodern}
\usepackage{amsmath}
\usepackage{amsfonts}
\usepackage{amssymb}
\usepackage{amsthm}
\usepackage[donotfixamsmathbugs,allowspaces]{mathtools}
\mathtoolsset{mathic}
\usepackage[pdfusetitle,pdfauthor={Giuseppe Sellaroli}]{hyperref}
\usepackage[all]{hypcap}
\usepackage[left=1.00in, right=1.00in, top=1.00in, bottom=1.00in]{geometry}
\usepackage{booktabs}
\usepackage{setspace}

\usepackage[affil-it]{authblk}
\usepackage[font={small}]{caption}
\usepackage[]{subfig}

\usepackage{algorithm}
\usepackage{algpseudocode}

\usepackage[capitalise]{cleveref}
\crefname{equation}{equation}{equations}

\usepackage[style=phys,biblabel=brackets,eprint]{biblatex}
\usepackage{csquotes}
\usepackage{tikz}
\usetikzlibrary{external,calc}
\pgfkeys{/pgf/images/include external/.code=\includegraphics{#1}} 

\tikzexternaldisable

\renewcommand{\vec}{\mathbf}

\newcommand{\pder}[2]{\frac{\partial{#1}}{\partial{#2}}}

\DeclareMathOperator{\Span}{span}

\DeclareMathOperator*{\argmax}{arg\,max}
\DeclareMathOperator*{\argmin}{arg\,min}

\newcommand{\casesif}{\quad\textnormal{if }\,} 
\newcommand{\casestextn}[1]{\quad\textnormal{#1}} 
\makeatletter
\def\setst{\@ifstar\@setst\@@setst} 
\newcommand{\@setst}{\:\middle\vert\:}
\newcommand{\@@setst}[1][]{\:#1\vert\:} 
\makeatother

\newcommand{\R}{\mathbb{R}}

\renewcommand{\S}{\mathbb{S}}

\makeatletter
\let\save@mathaccent\mathaccent
\newcommand*\if@single[3]{%
  \setbox0\hbox{\({\mathaccent"0362{#1}}^H\)}%
  \setbox2\hbox{\({\mathaccent"0362{\kern0pt#1}}^H\)}%
  \ifdim\ht0=\ht2 #3\else #2\fi
  }
\newcommand*\rel@kern[1]{\kern#1\dimexpr\macc@kerna}
\newcommand*\widebar[1]{\@ifnextchar^{\wide@bar{#1}{0}}{\wide@bar{#1}{1}}}
\newcommand*\wide@bar[2]{\if@single{#1}{\wide@bar@{#1}{#2}{1}}{\wide@bar@{#1}{#2}{2}}}
\newcommand*\wide@bar@[3]{%
  \begingroup
  \def\mathaccent##1##2{%
    \let\mathaccent\save@mathaccent
    \if#32 \let\macc@nucleus\first@char \fi
    \setbox\z@\hbox{\(\macc@style{\macc@nucleus}_{}\)}%
    \setbox\tw@\hbox{\(\macc@style{\macc@nucleus}{}_{}\)}%
    \dimen@\wd\tw@
    \advance\dimen@-\wd\z@
    \divide\dimen@ 3
    \@tempdima\wd\tw@
    \advance\@tempdima-\scriptspace
    \divide\@tempdima 10
    \advance\dimen@-\@tempdima
    \ifdim\dimen@>\z@ \dimen@0pt\fi
    \rel@kern{0.6}\kern-\dimen@
    \if#31
      \overline{\rel@kern{-0.6}\kern\dimen@\macc@nucleus\rel@kern{0.4}\kern\dimen@}%
      \advance\dimen@0.4\dimexpr\macc@kerna
      \let\final@kern#2%
      \ifdim\dimen@<\z@ \let\final@kern1\fi
      \if\final@kern1 \kern-\dimen@\fi
    \else
      \overline{\rel@kern{-0.6}\kern\dimen@#1}%
    \fi
  }%
  \macc@depth\@ne
  \let\math@bgroup\@empty \let\math@egroup\macc@set@skewchar
  \mathsurround\z@ \frozen@everymath{\mathgroup\macc@group\relax}%
  \macc@set@skewchar\relax
  \let\mathaccentV\macc@nested@a
  \if#31
    \macc@nested@a\relax111{#1}%
  \else
    \def\gobble@till@marker##1\endmarker{}%
    \futurelet\first@char\gobble@till@marker#1\endmarker
    \ifcat\noexpand\first@char A\else
      \def\first@char{}%
    \fi
    \macc@nested@a\relax111{\first@char}%
  \fi
  \endgroup
}
\makeatother

\theoremstyle{plain}
\newtheorem{proposition}{Proposition}

\newtheorem{theorem}{Theorem}
\newtheorem*{theorem*}{Theorem}

\theoremstyle{definition}

\DeclarePairedDelimiter\abs{\lvert}{\rvert}
\DeclarePairedDelimiter\norm{\lVert}{\rVert}

\DeclarePairedDelimiter\paren{\lparen}{\rparen}

\DeclarePairedDelimiter\bracks{\lbrack}{\rbrack}

\DeclarePairedDelimiterXPP{\eval}[2]{}{.}{\rvert}{_{#2}}{#1}
\makeatletter
\let\oldeval\eval
\def\eval{\@ifstar{\oldeval}{\oldeval*}}
\makeatother
\let\set\relax
\DeclarePairedDelimiter\set{\lbrace}{\rbrace}

\let\braket\relax
\DeclarePairedDelimiter{\braket}{\langle}{\rangle}

\let\originalleft\left
\let\originalright\right
\renewcommand{\left}{\mathopen{}\mathclose\bgroup\originalleft}
\renewcommand{\right}{\aftergroup\egroup\originalright}

\makeatletter
\newcommand\fs@booktabsruled{%
  \def\@fs@cfont{\bfseries\strut}\let\@fs@capt\floatc@ruled
  \def\@fs@pre{\hrule height\heavyrulewidth depth0pt \kern\belowrulesep}%
  \def\@fs@mid{\kern\aboverulesep\hrule height\lightrulewidth\kern\belowrulesep}%
  \def\@fs@post{\kern\aboverulesep\hrule height\heavyrulewidth\relax}%
  \let\@fs@iftopcapt\iftrue
}
\newcounter{algorithmicH}
\let\oldalgorithmic\algorithmic
\renewcommand{\algorithmic}{%
  \stepcounter{algorithmicH}
  \oldalgorithmic}
\renewcommand{\theHALG@line}{ALG@line.\thealgorithmicH.\arabic{ALG@line}}
\makeatother

\floatstyle{booktabsruled}
\restylefloat{algorithm}

\usepackage{enumitem}
\newlist{proofenum}{enumerate}{1}
\setlist[proofenum,1]{wide, labelwidth=!, labelindent=0pt,label={(\roman*)}}
\newenvironment{proofenumerate}{\begin{proofenum}}{\qedhere\end{proofenum}}

\title{An algorithm to reconstruct convex polyhedra
\\
from their face normals and areas}
\author[a]{Giuseppe Sellaroli}
\affil[a]{University of Waterloo}
\date{\today}

\newcommand{\softwareurl}{\href{https://github.com/gsellaroli/polyhedrec}{https://github.com/gsellaroli/polyhedrec}}

\begin{document}
\begin{titlepage}
\makeatletter
\centering
\begin{doublespace}
{\bfseries \Large \@title}
\end{doublespace}
\vspace{\baselineskip}\par
\large\@author
\vspace{\baselineskip}\par
\normalfont\normalsize\@date
\vspace{1.5\baselineskip}\par
\makeatother
\begin{abstract}
A well-known result in the study of convex polyhedra, due to Minkowski, is that a convex polyhedron is uniquely determined (up to translation) by the directions and areas of its faces. The theorem guarantees existence of the polyhedron associated to given face normals and areas, but does not provide a constructive way to find it explicitly. This article provides an algorithm to reconstruct 3D convex polyhedra from their face normals and areas, based on an method by Lasserre to compute the volume of a convex polyhedron in \( \R^n \).
A \emph{Python} implementation of the algorithm is available at {\softwareurl}.
\end{abstract}
\end{titlepage}

\section{Introduction}

Minkowski's theorem is a well-known result in the theory of convex polyhedra, stating that a bounded convex polyhedron is uniquely determined by its outward face normals and its face areas, up to translation. However, the theorem does not provide a procedure to explicitly reconstruct the polyhedron associated to face normals and areas.
The results of Minkowski's theorem are very important, and are for example widely used in \emph{loop quantum gravity}, where vectors encoding both face normals and areas are used to represent geometry through the convex polyhedra corresponding to them~\cite{Rovelli2009,BianchiDonaSpeziale2011}.
Using face normals and areas as variables to describe polyhedra has the disadvantage that it is generally not obvious what the polyhedron actually looks like. This article overcomes this obstacle by introducing a numerical algorithm, henceforth referred to as \emph{polyhedrec}, to reconstruct the polyhedron from normals and areas. Pseudo-code for the algorithm is provided, and a \emph{Python} implementation has been made available at {\softwareurl}.

The results of the present paper make use of an algorithm by Lasserre to compute the volume of a polyhedron in \( \R^n \)~\cite{Lasserre1983}, here used to compute areas of convex polygons. The idea behind \emph{polyhedrec} is the following: given \( F \) face normals and areas, it considers a class of convex polyhedra with the given face normals (parametrised by \( F \) real numbers \( h_i \)), then it finds the parameters \( \vec{h}=(h_1,\dotsc,h_F) \) for which the face areas of the polyhedron \( P(\vec{h}) \) match the given areas.

It should be noted that a reconstruction technique similar to \emph{polyhedrec} was already introduced by Bianchi, Doná, and Speziale~\cite{BianchiDonaSpeziale2011}, and indeed this article inspired the use of Lasserre's algorithm here. Despite the similarity with this earlier work, the algorithm presented in the following has several advantages:
\begin{itemize}
\item both algorithms work by to finding the roots of a vector function, but \emph{polyhedrec} uses a function with a single root, while the function in  ref.~\cite{BianchiDonaSpeziale2011} has three degrees of freedom, resulting in infinitely many roots;
\item \emph{polyhedrec} provides an explicit expression for the Jacobian matrix of the function, improving the convergence of the root-finding algorithm;
\item \emph{polyhedrec} provides pseudo-code and an implementation of the algorithm, while many of the steps in ref.~\cite{BianchiDonaSpeziale2011} are not made explicit\footnote{It should be noted that the reconstruction algorithm is not the main topic of ref.~\cite{BianchiDonaSpeziale2011}, which focuses on the applications of polyhedra in loop quantum gravity.};
\item there are some oversights and erroneous statements in ref.~\cite{BianchiDonaSpeziale2011} that are corrected here.
\end{itemize}

The article is structured as follows. \Cref{sec:premise} provides a brief review of Minkowski's theorem and introduces notation. The actual algorithm is presented in \cref{sec:reconstruction}; the section is divided in three parts, focusing respectively on finding the face areas of a generic convex polyhedron, computing the Jacobian matrix of the the face area function, and finding the reconstructed polyhedron. Finally, some concluding remarks are given in \cref{sec:conclusions}.

\section{Premise}\label{sec:premise}

We define a convex polyhedron as the region of \( P\subseteq\R^3 \) obtained as the intersection of finitely many \emph{half-spaces}
\begin{equation}
H_i = \set{\vec{x}\in\R^3 \setst \braket{\vec{u}_i,\vec{x}}\leq h_i},\quad \vec{u}_i\in \S^2,\quad h_i\in \R,
\end{equation}
where \( \braket{\cdot,\cdot} \) denotes the Euclidean inner product and \( \S^2 \) is the \( 2 \)-sphere
\begin{equation}
\S^2 = \set{\vec{x}\in \R^3 \setst \norm{\vec{x}}=1}.
\end{equation}
There are obviously infinitely many ways to represent a given polyhedron in terms of half-spaces; the representation is called \emph{minimal} if intersection of each plane
\begin{equation}
Q_i = \partial H_i = \set{\vec{x}\in \R^3 \setst \braket{\vec{u}_i,\vec{x}}= h_i}
\end{equation}
and the polyhedron is a \( 2 \)D subset\footnote{Here we mean that the region \( Q_i \cap P \) has non-zero area.} of \( \R^3 \).
When \( \dim(Q_i\cap P) = 2 \) we are going to call the regions \( f_i := Q_i \cap P \) \emph{faces} of the polyhedron, while the regions \( e_{ij}:= Q_i \cap Q_j \cap P \) are going to be called \emph{edges} when \( \dim(e_{ij})=1 \).

One can easily see that the vectors \( \vec{u}_i \) are the normals to the planes \( Q_i \); moreover, they point outside of the half-space \( H_i \), so we are going to call them \emph{outward normals}. Whenever \( f_i \) is a face we refer to \( \vec{u}_i \) as its \emph{outward face normal}.
Under certain conditions, polyhedra are fully determined by their face normals and and face areas, as proved by Minkowski~\cite{Alexandrov2005}.
\begin{theorem}[Minkoswki]\label{thm:minkowski}
Let \( \set{\vec{u}_i}_{i=1}^F \subset \S^2\) and \( \set{A_i}_{i=1}^F \subset \R \) be respectively a set of \emph{distinct} unit vectors spanning \( \R^3 \) and a set of strictly positive numbers, satisfying the closure constraint
\begin{equation*}
\sum_i A_i\vec{u}_i = 0.
\end{equation*}
Then there is a bounded convex polyhedron \( P \) with outward face normals \( \vec{u}_i \) and face areas \( A_i \), unique up to translation.
\end{theorem}
In other words, the theorem tells us that, given \( \vec{u}_i \) and \( A_i \) satisfying its hypotheses, there are numbers \(h_i\) such that \( \vec{u}_i \) are
the outward face normals to the polyhedron
\begin{equation}
P = \set{\vec{x} \in \R^3 \setst \braket{\vec{u}_i,\vec{x}} \leq h_i, i =1,\dotsc,F}
\end{equation}
and each \( f_i \) has area \( A_i \). The numbers \( h_i \) are not unique, as they are defined up to a translation
\begin{equation}
\vec{x}  \rightarrow \vec{x} + \vec{c},\quad \vec{c}\in \R^3,
\end{equation}
under which they transform as
\begin{equation}
h_i \rightarrow h_i + \braket{\vec{u}_i,\vec{c}}.
\end{equation}
The theorem is not constructive, so it only tells us that the numbers \( h_i \) exist, not what they are. The goal of this paper is to build an algorithm to reconstruct the actual polyhedron from its unit normals and areas, assuming they satisfy the hypotheses of Minkowski's theorem.

\section{Reconstructing the polyhedron}\label{sec:reconstruction}

Let \( \set{\vec{u}_i}_{i=1}^F \) be the outward unit normals of the polyhedron we want to reconstruct and \( \set{A^0_i}_{i=1}^F \) the respective face areas, satisfying the assumptions of Minkoswki's theorem. If we denote by
\begin{equation}
P(\vec{h}) := \set{\vec{x} \in \R^3 \setst \braket{\vec{u}_i,\vec{x}} \leq h_i, i =1,\dotsc,F}
\end{equation}
the generic polyhedron with outward normals \( \vec{u}_i \), reconstructing the polyhedron is equivalent to finding \( \vec{h} \) such that
\begin{equation}
\label{eq:area constraint}
A_i(\vec{h}) : = \operatorname{area}(f_i) = A^0_i.
\end{equation}
Our reconstruction algorithm is going to follow this approach:
\begin{itemize}
\item find an expression for the functions \( A_i(\vec{h}) \);
\item compute the Jacobian \( J_{ij}(\vec{h}) = \pder{A_i}{h_j} \);
\item use a root-finding algorithm to find a solution to \cref{eq:area constraint}.
\end{itemize}

\subsection{Computing the area of \texorpdfstring{\( f_i(\vec{h}) \)}{fi(h)}}

Our algorithm to compute the area of each \( f_i(\vec{h}) \) is based on the one presented in \cite{BianchiDonaSpeziale2011}, which itself is based on and algorithm by Lasserre~\cite{Lasserre1983}: the idea is to compute the lengths of the regions \( e_{ij} \), \( j\neq i \) which bound \( f_i \) and use them to calculate the area. Although the underlying idea is the same as \cite{BianchiDonaSpeziale2011}, the presentation here will be more geometrical and avoid many unnecessary variables; most importantly, there are some oversights in \cite{BianchiDonaSpeziale2011} that will be addressed in the present paper.

Before we can devise a way to compute the face areas of \( P(\vec{h}) \), there is an important point to consider. Although we know from Minkowski's theorem that the reconstructed polyhedron is bounded, there is \emph{a priori} no guarantee that \( P(\vec{h}) \) is for arbitrary \( \vec{h} \), which could result in infinite areas. Luckily, \( P(\vec{h}) \) is guaranteed to be bounded by the following \namecref{prop:P(h) is bounded}.
\begin{proposition}\label{prop:P(h) is bounded}
Let \( \set{\vec{u}_i}_{i=1}^F \subset \S^2 \) be unit vectors spanning \( \R^3 \), and suppose there are numbers \( \alpha_i >0 \) such that
\begin{equation*}
\sum_i \alpha_i \vec{u}_i = 0.
\end{equation*}
Then the polyhedron
\(P(\vec{h})= \set{\vec{x} \in \R^3 \setst \braket{\vec{u}_i,\vec{x}} \leq h_i, i =1,\dotsc,F}\)
is bounded for any \( \vec{h}\in \R^F \).
\end{proposition}
\begin{proof}
We are going to consider three separate cases.
\begin{proofenumerate}
\item First note that if \( P(\vec{h}) = \emptyset \) the proposition holds, as the empty set is bounded.
\item\label{proof:non-empty interior} Suppose now that \( P(\vec{h}) \setminus \bigcup_i Q_i(\vec{h}) \neq \emptyset\), i.e., the interior of the polyhedron is non-empty, and let \( \vec{p}\) be a point in it. The ray
\begin{equation}
r(\vec{u}):= \set{\vec{p} + \lambda \vec{u} \setst \lambda \geq 0}
\end{equation}
in the direction \( \vec{u}\in \S^2 \) intersects the plane \( Q_i \) if and only if
\begin{equation}
\exists \lambda \geq 0 \quad\text{s.t.}\quad\braket{\vec{u}_i,\lambda\vec{u}} = h_i - \braket{\vec{u}_i,\vec{p}}>0,
\end{equation}
or equivalently iff
\begin{equation}
\braket{\vec{u}_i,\vec{u}} > 0.
\end{equation}
Suppose that \( P(\vec{h}) \) is unbounded; then there is \( \vec{u}\in \S^2 \) such that \( r(\vec{u}) \) does not intersect any of the \( Q_i \), i.e.,
\begin{equation}
\braket{\vec{u}_i,\vec{u}}\leq 0, \quad \forall i.
\end{equation}
However,
\begin{equation}
\vec{u}_F = - \sum_{i\neq F} \beta_i \vec{u}_i,\quad \beta_i : =\frac{\alpha_i}{\alpha_F}>0,
\end{equation}
so that
\begin{equation}
\braket{\vec{u}_F,\vec{u}} =  -\sum_{i\neq F} \beta_i \braket{\vec{u}_i,\vec{u}} \geq 0.
\end{equation}
Moreover, since \( \Span\set{\vec{u}_i}_{i=1}^F = \R^3\) and \( \norm{\vec{u}}=1 \) there must be at least one \( \vec{u}_i \) such that \( \braket{\vec{u}_i,\vec{u}}\neq 0 \); whether the latter is \( \vec{u}_F \) or not, this leads to
\begin{equation}
\braket{\vec{u}_F,\vec{u}} > 0,
\end{equation}
which is a contradiction.
\item Finally, when \( P(\vec{h}) \setminus \bigcup_i Q_i(\vec{h}) = \emptyset\) but \( P(\vec{h}) \neq \emptyset \) note that
\begin{equation}
\emptyset \neq P(\vec{h}) \subseteq P(\vec{h'}) \setminus \bigcup_i Q_i(\vec{h'}),\quad h'_i = h_i + 1
\end{equation}
which implies \(  P(\vec{h}')  \) is bounded as a consequence of point \ref{proof:non-empty interior}. It follows that \( P(\vec{h}) \) is bounded as well.
\end{proofenumerate}
\end{proof}

Armed with the results of \cref{prop:P(h) is bounded}, we can proceed in our computation of the area of \( f_i(\vec{h}) \). To do so, we are first going to find the lengths of all the edges \( e_{ij} \), \( j\neq i \) bounding \( f_i \). We distinguish two cases.

\subsubsection*{Length of \( e_{ij} \), case 1: \( \vec{u}_j \neq -\vec{u}_i \)}
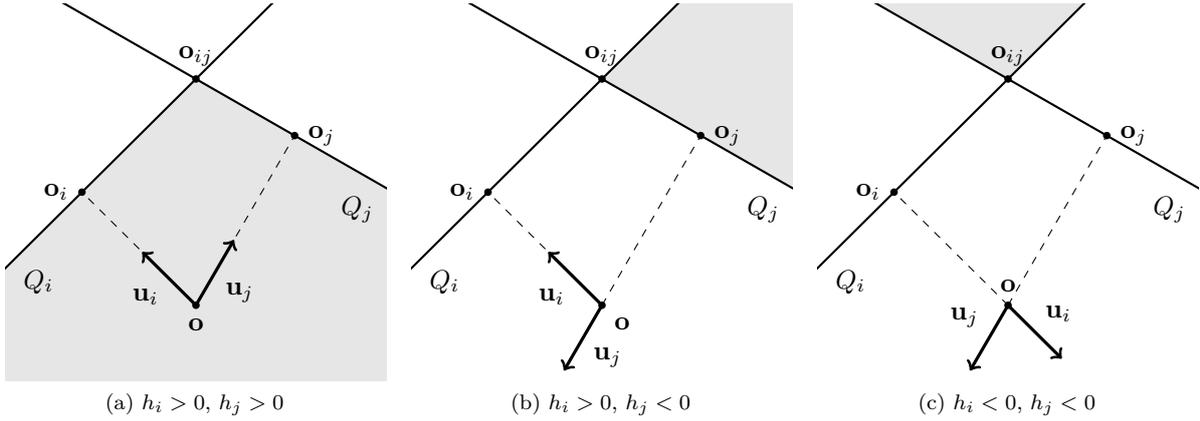
\begin{figure}[!th]
\centering
\subfloat[\( h_i>0\), \(h_j>0 \)]{\label{subfig:both positive}
\begin{tikzpicture}
\coordinate (o) at (0,0);
\coordinate (o_ij) at (90:3);
\coordinate (o_i) at (135:2.121);
\coordinate (o_j) at (60:2.598);

\clip[] (-2.5,-1) rectangle (2.5,4);

\fill[black!10!white,shift=(o_ij)] (0:0) -- (225:10cm) -- (-30:10cm) -- cycle;

\draw[dashed] (0:0) -- (o_i);
\draw[->, very thick] (0:0) -- (135:1) node[pos=0.5,below left] {\( \vec{u}_i \)};
\draw[thick,shorten < = -5cm, shorten > = -5cm] (o_i) -- (o_ij) node[pos=-0.6,below right] {\( Q_i \)};
\node[fill,circle,inner sep=1pt,label={left:\( \vec{o}_i \)}] at (o_i) {};

\draw[dashed] (0:0) -- (o_j);
\draw[->, very thick] (0:0) -- (60:1) node[pos=0.5,below right] {\( \vec{u}_j \)};
\draw[thick,shorten < = -5cm, shorten > = -5cm] (o_j) -- (o_ij) node[pos=-0.9,below left] {\( Q_j \)};
\node[fill,circle,inner sep=1pt,label={right:\( \vec{o}_j \)}] at (o_j) {};

\node[fill,circle,inner sep=1pt,label={below:\( \vec{o} \)}] at (o) {};

\node[fill,circle,inner sep=1pt,label={above:\( \vec{o}_{ij} \)}] at (o_ij) {};
\end{tikzpicture}
}
\hfill
\subfloat[\( h_i>0\), \(h_j<0 \)]{\label{subfig:positive + negative}
\begin{tikzpicture}
\coordinate (o) at (0,0);
\coordinate (o_ij) at (90:3);
\coordinate (o_i) at (135:2.121);
\coordinate (o_j) at (60:2.598);

\clip[] (-2.5,-1) rectangle (2.5,4);

\fill[black!10!white,shift=(o_ij)] (0:0) -- (225:-10cm) -- (-30:10cm) -- cycle;

\draw[dashed] (0:0) -- (o_i);
\draw[->, very thick] (0:0) -- (135:1) node[pos=0.5,below left] {\( \vec{u}_i \)};
\draw[thick,shorten < = -5cm, shorten > = -5cm] (o_i) -- (o_ij) node[pos=-0.6,below right] {\( Q_i \)};
\node[fill,circle,inner sep=1pt,label={left:\( \vec{o}_i \)}] at (o_i) {};

\draw[dashed] (0:0) -- (o_j);
\draw[->, very thick] (0:0) -- (60:-1) node[pos=0.5,below right] {\( \vec{u}_j \)};
\draw[thick,shorten < = -5cm, shorten > = -5cm] (o_j) -- (o_ij) node[pos=-0.9,below left] {\( Q_j \)};
\node[fill,circle,inner sep=1pt,label={right:\( \vec{o}_j \)}] at (o_j) {};

\node[fill,circle,inner sep=1pt,label={below right:\( \vec{o} \)}] at (o) {};

\node[fill,circle,inner sep=1pt,label={above:\( \vec{o}_{ij} \)}] at (o_ij) {};
\end{tikzpicture}
}
\hfill
\subfloat[\( h_i<0\), \(h_j<0 \)]{\label{subfig:both negative}
\begin{tikzpicture}
\coordinate (o) at (0,0);
\coordinate (o_ij) at (90:3);
\coordinate (o_i) at (135:2.121);
\coordinate (o_j) at (60:2.598);

\clip[] (-2.5,-1) rectangle (2.5,4);

\fill[black!10!white,shift=(o_ij)] (0:0) -- (225:-10cm) -- (-30:-10cm) -- cycle;

\draw[dashed] (0:0) -- (o_i);
\draw[->, very thick] (0:0) -- (135:-1) node[pos=0.5,above right] {\( \vec{u}_i \)};
\draw[thick,shorten < = -5cm, shorten > = -5cm] (o_i) -- (o_ij) node[pos=-0.6,below right] {\( Q_i \)};
\node[fill,circle,inner sep=1pt,label={left:\( \vec{o}_i \)}] at (o_i) {};

\draw[dashed] (0:0) -- (o_j);
\draw[->, very thick] (0:0) -- (60:-1) node[pos=0.5,above left] {\( \vec{u}_j \)};
\draw[thick,shorten < = -5cm, shorten > = -5cm] (o_j) -- (o_ij) node[pos=-0.9,below left] {\( Q_j \)};
\node[fill,circle,inner sep=1pt,label={right:\( \vec{o}_j \)}] at (o_j) {};

\node[fill,circle,inner sep=1pt,label={above:\( \vec{o} \)}] at (o) {};

\node[fill,circle,inner sep=1pt,label={above:\( \vec{o}_{ij} \)}] at (o_ij) {};
\end{tikzpicture}
}
\caption{Intersection of the planes \( Q_i \) and \( Q_j \) as seen from the plane spanned by \( \vec{u}_i \), \( \vec{u}_j \). Note that \( \vec{o} \) (the origin), \( \vec{o}_i \), \( \vec{o}_j \), and \( \vec{o}_{ij} \) all belong to this plane.  Different possibilities for the signs of \( h_i \) and \( h_j \) are shown; the darker area is the intersection of the half-spaces \( H_i \) and \( H_j \).}
\label{fig:span(u_i u_j)}
\end{figure}%

When \( \vec{u}_j \neq -\vec{u}_i \) the planes \( Q_i \) and \( Q_j \) are not parallel\footnote{Recall that the unit normals are required to be distinct, so it cannot be \( \vec{u}_j = \vec{u}_i \).}, hence the intersection \( Q_i \cap Q_j \) is a line.
Let \( \vec{o}_i \) be the projection of the origin on the plane \( Q_i \), i.e., 
\begin{equation}
\vec{o}_i = h_i \vec{u_i}.
\end{equation}
Note that \( \norm{\vec{o}-\vec{o}_i} = \abs{h_i} \), so that \( h_i \) gives us the \emph{signed distance} between the origin and \( Q_i \); in general, we are going to define the signed distance between a point \( \vec{x} \) and one of the planes \( Q_i \) as
\begin{equation}
\sigma(\vec{x},Q_i):= h_i - \braket{\vec{u}_i,\vec{x}},
\end{equation}
so that
\begin{equation}
\begin{cases}
\sigma(\vec{x},Q_i) \geq 0 & \casesif \vec{x}\in H_i\\
\sigma(\vec{x},Q_i)  < 0 & \casestextn{otherwise}.
\end{cases}
\end{equation}
We denote by \( \vec{o}_{ij} \) the intersection between \( Q_i \cap Q_j \) and  the plane containing the origin spanned by \( \vec{u}_i \) and \( \vec{u}_j \) (see \cref{fig:span(u_i u_j)}), that is
\begin{equation}
\begin{cases}
\vec{o}_{ij} = \alpha \vec{u}_i + \beta \vec{u}_j,\quad \alpha,\beta\in\R\\
\braket{\vec{u}_i,\vec{o}_{ij}}=h_i\\
\braket{\vec{u}_j,\vec{o}_{ij}}=h_j,
\end{cases}
\end{equation}
whose solution is
\begin{equation}
\vec{o}_{ij} = \frac{r_{ji} \vec{u}_i + r_{ij} \vec{u}_j}{1-\braket{\vec{u}_i,\vec{u}_j}^2} \equiv \frac{r_{ji} \vec{u}_i + r_{ij} \vec{u}_j}{\norm{\vec{u}_i \times\vec{u}_j}^2},
\quad
r_{ij} : = h_j - \braket{\vec{u}_i,\vec{u}_j} h_i.
\end{equation}
Since the line \( Q_i \cap Q_j \) is perpendicular to both \( \vec{u}_i \) and \( \vec{u}_j \) it has \( \vec{u}_i \times \vec{u}_j\) as a direction vector, i.e.,
\begin{equation}
Q_i \cap Q_j = \set{\vec{o}_{ij} + \lambda\vec{u}_i \times \vec{u}_j \setst \lambda \in \R}.
\end{equation}
Each point \( \vec{o}_{ij} + \lambda\vec{u}_i \times \vec{u}_j \) belongs to \( e_{ij} \) if and only if
\begin{equation}
\braket{\vec{u}_k,\vec{o}_{ij} + \lambda\vec{u}_i \times \vec{u}_j} \leq h_k,\quad \forall k \neq i,j,
\end{equation}
that is
\begin{equation}
\label{eq: a-lambda-b inequality}
a^k_{ij} \lambda \leq b^k_{ij},\quad \forall k \neq i,j,
\end{equation}
where
\begin{equation}
a^k_{ij} := \braket{\vec{u}_k, \vec{u}_i \times \vec{u}_j},
\quad
b^k_{ij}:= h_k - \frac{r_{ij} \braket{\vec{u}_j,\vec{u}_k} + r_{ji} \braket{\vec{u}_i,\vec{u}_k}}{\norm{\vec{u}_i \times \vec{u}_j}^2};
\end{equation}
it follows that
\begin{equation}
\lambda \leq \lambda^\text{max}_{ij} := \min_{k} \set[\bigg]{\frac{b^k_{ij}}{a^k_{ij}} \setst[\bigg] a^k_{ij}>0}
\quad\text{and}\quad
\lambda \geq \lambda^\text{min}_{ij} := \max_{k} \set[\bigg]{\frac{b^k_{ij}}{a^k_{ij}} \setst[\bigg] a^k_{ij}<0}.
\end{equation}
Note however that we still have to take into account \cref{eq: a-lambda-b inequality} when \( a^k_{ij} =0\): in fact, in this case the equation is either solved by every lambda (if \( b^k_{ij}\geq 0 \)) or by none of them  (if \( b^k_{ij}< 0 \)). This point is ignored in~\cite{BianchiDonaSpeziale2011}, where only the case \( a^k_{ij} \neq 0\) is used to compute the length, leading to an erroneous calculation is some cases.

Putting everything together, we can write
\begin{equation}
e_{ij} = \set{\vec{o}_{ij} + \lambda\, \vec{u}_i \times \vec{u}_j \setst \lambda \in \Lambda_{ij}},
\end{equation}
where
\begin{equation}
\Lambda_{ij} = \bracks{\lambda^\text{min}_{ij},\lambda^\text{max}_{ij}} \cap  \bigcap_{k} \set{\lambda \in \R \setst a^k_{ij}=0, b^k_{ij}\geq 0}
\end{equation}
and we use the convention that \( \bracks{a,b}=\emptyset  \) if \( a>b \). When \( \Lambda_{ij} \neq \emptyset \) the length of \( e_{ij} \) is
\begin{equation}
L_{ij} = \norm{(\vec{o}_{ij} + \lambda^\text{max}_{ij}\, \vec{u}_i \times \vec{u}_j) - (\vec{o}_{ij} + \lambda^\text{min}_{ij}\, \vec{u}_i \times \vec{u}_j)}
=
(\lambda^\text{max}_{ij} - \lambda^\text{min}_{ij}) \norm{\vec{u}_i \times \vec{u}_j},
\end{equation}
so that in general
\begin{equation}
L_{ij} = \lambda_{ij} \norm{\vec{u}_i \times \vec{u}_j},\quad
\lambda_{ij} :=
\begin{cases}
\lambda^\text{max}_{ij} - \lambda^\text{min}_{ij} & \casesif \Lambda_{ij}\neq \emptyset
\\
0 & \casestextn{otherwise}.
\end{cases}
\end{equation}

\subsubsection*{Length of \( e_{ij} \), case 2: \( \vec{u}_j = -\vec{u}_i \)}

Suppose now that \( \vec{u}_j = -\vec{u}_i \), which implies \( Q_i \) and \( Q_j \) are parallel. The intersection \( Q_i \cap Q_j \) in this case is either empty (\( Q_i \neq Q_j \)) or a plane (\( Q_i = Q_j \)). In either case \( e_{ij} \) is not an edge, so we set
\begin{equation}
L_{ij} = \lambda_{ij} = 0.
\end{equation}
\subsubsection*{Using the lengths to find the area}
Now that we have the lengths of the edges bounding \( f_i \) (if any), we can use them to compute its area as follows. Given a convex polygon---which \( f_i \) is, if it is a face---such as in \cref{fig:area with triangles}, its area can be computed in the following way:
\begin{figure}[ht]
\centering
\begin{tikzpicture}[thick,scale=1.5]
\coordinate (a) at (0:0);
\coordinate (b) at (30:1.5);
\coordinate (c) at (-10:3);
\coordinate (d) at (-40:2.5);
\coordinate (o) at (5:5);
\coordinate (midpoint) at ($(b)!0.5!(d)$);
\node[right] at (o) {\( O \)};
\node[left] at (a) {\( A \)};
\node[above] at (b) {\( B \)};
\node[right] at (c) {\( C \)};
\node[below] at (d) {\( D \)};

\filldraw (o) circle (1pt);
\filldraw[fill=black!10!white] (a) -- (b) -- (c) -- (d) -- cycle;
\draw[dashed] (a) -- (o);
\draw[dashed] (b) -- (o);
\draw[dashed] (c) -- (o);
\draw[dashed] (d) -- (o);
\end{tikzpicture}
\\[\baselineskip]
\( \operatorname{area}(ABCD) =  \operatorname{area}(ABO) + \operatorname{area}(ADO) - \operatorname{area}(DCO) - \operatorname{area}(BCO)\)
\caption{Calculation of the area of a convex polygon using triangles. Note that if we choose the common vertex of the triangles \( O \) to be outside the polygon, some of the triangle areas have to be subtracted, namely those for which \( O \) is on the other side of the edge with respect to the polygon.}
\label{fig:area with triangles}
\end{figure}
\begin{itemize}
\item choose any point \( O \) on the plane containing the polygon;
\item for each edge of the polygon, construct the triangle with the endpoints of the edge and \( O \) as vertices.
\end{itemize}
Note that by construction the triangles cover the polygon. There are two possible cases, depending on the position of \( O \).
\begin{enumerate}
\item If \( O \) is inside the polygon, the triangles do not overlap and their union coincides with the polygon itself. It follows that the area of the polygon can be obtained by summing the area of each triangle.
\item If \( O \) is outside the polygon, some of the triangles overlap. Specifically, there are two kinds of triangles, i.e., those for which \( O \) and the polygon lie on the same side the associated edge (e.g., \( ABO \) in the figure), and those for which they lie on opposite sides (e.g., \( BCO \)); we will call them respectively of type I and II. One can easily see that the area of the polygon is obtained by summing the areas of all the type I triangles and subtracting those of all the type II triangles.
\end{enumerate}
We are going to compute the area \( A_i(\vec{h}) \) using this procedure, with \( \vec{o}_i \) as our chosen point on the plane \( Q_i \). The area of the triangle associated to \( e_{ij} \) is
\begin{equation}
\frac{1}{2} L_{ij} d_{ij},
\end{equation}
where \( d_{ij} \) is the distance between \( \vec{o}_i \) and \( Q_i \cap Q_j \), i.e, the distance between \( \vec{o}_i \) and \( \vec{o}_{ij} \), as
\begin{equation}
\braket{\vec{o}_{ij} - \vec{o}_i, \vec{u}_i \times \vec{u}_j} = 0.
\end{equation}
To account for the possibility of the triangle being of type II, we will make \( d_{ij} \) a signed distance, with
\begin{equation}
\begin{cases}
d_{ij} \geq 0 & \casesif \vec{o}_i \in H_j\\
d_{ij} < 0 & \casestextn{otherwise}.
\end{cases}
\end{equation}
One can easily see that
\begin{equation}
\vec{o}_{ij} - \vec{o}_i = r_{ij}\frac{\vec{u}_j - \braket{\vec{u}_i,\vec{u}_j} \vec{u}_i}{\norm{\vec{u}_i \times \vec{u}_j}^2},
\end{equation}
which since
\begin{equation}
\norm{\vec{u}_j - \braket{\vec{u}_i,\vec{u}_j} \vec{u}_i}^2 =
1 - \braket{\vec{u}_i,\vec{u}_j}^2 = \norm{\vec{u}_i \times \vec{u}_j}^2
\end{equation}
implies that
\begin{equation}
\norm{\vec{o}_{ij} - \vec{o}_i} = \frac{\abs{r_{ij}}}{\norm{\vec{u}_i \times \vec{u}_j}}.
\end{equation}
Noticing that
\begin{equation}
r_{ij} = h_j - \braket{\vec{u}_i,\vec{u}_j} h_i \geq 0
\quad\Leftrightarrow\quad
\braket{\vec{u}_j, \vec{o}_i} \leq h_j
\quad\Leftrightarrow\quad
\vec{o}_i \in H_j,
\end{equation}
we see that it must be
\begin{equation}
d_{ij} = \frac{r_{ij}}{\norm{\vec{u}_i \times \vec{u}_j}};
\end{equation}
the area of \( f_i \) can finally be calculated as\footnote{The case where \( \vec{u}_j = -\vec{u}_i \), in which \( d_{ij} \) is not defined, is accounted for since \( L_{ij}=\lambda_{ij}=0 \).}
\begin{equation}
A_i(\vec{h}) = \sum_{j\neq i} \frac{1}{2} L_{ij} d_{ij} \equiv \sum_{j\neq i} \frac{1}{2} \lambda_{ij} r_{ij}.
\end{equation}

Pseudocode for the calculation of all the face areas of \( P(\vec{h}) \) is presented in \cref{alg:areas} as a summary. In this form we can easily see that the time complexity for the computation of the areas is \( O(F^3) \). Since it is obviously
\begin{equation}
\lambda_{ji} = \lambda_{ij},
\end{equation}
the algorithm can be made more efficient by only cycling through \( j>i \), which is indeed what was done in the Python implementation.
\begin{algorithm}[!htbp]
\caption{Face areas of \( P(\vec{h}) \)}
\label{alg:areas}
\begin{algorithmic}[1]
\Procedure{$\text{Areas}$}{$\vec{h}$}
\State \( A \gets (0,0,\dotsc,0) \)\Comment{Initialise areas array (sequence of \( n \) zeros)}
\State Compute all the \( \lambda_{ij} \):
\For {\( i\in\set{1,\dotsc,F} \)}
\For{$j \neq i$}
\If {\( \braket{\vec{u}_i,\vec{u}_j} = -1 \)}
\State \( \lambda_{ij} \gets 0 \)\Comment{If the \( f_i \) and \( f_j \) are parallel \( \lambda_{ij} \) must be zero.}
\Else
\State $\lambda^\text{max}_{ij}\gets \min_{a^k_{ij}>0} \set*{{b^k_{ij}(\vec{h})}/{a^k_{ij}}}$
\State $\lambda^\text{min}_{ij}\gets \max_{a^k_{ij}<0} \set*{{b^k_{ij}(\vec{h})}/{a^k_{ij}}}$
\State $\lambda_{ij} \gets \max(0,\lambda^\text{max}_{ij} - \lambda^\text{min}_{ij})$
\For{$k\neq i,j$}
\If {\( a^k_{ij} = 0 \) \textbf{and} \( b^k_{ij}(\vec{h})<0 \)}
\State \( \lambda_{ij} \gets 0 \)
\Comment{If \( a^k_{ij}=0 \) and \( b^k_{ij} <0\) for some \( k\neq i,j \) then \( \lambda_{ij} = 0 \).}
\EndIf
\EndFor
\EndIf
\EndFor
\EndFor
\State Compute areas:
\For {\( i\in\set{1,\dotsc,F} \)}
\For{$j \neq i$}
\State $A_i \gets A_i + \frac{1}{2} \lambda_{ij} r_{ij}(\vec{h})$
\EndFor
\EndFor
\State \textbf{return} $A$
\EndProcedure
\end{algorithmic}
\end{algorithm}

\subsection{The Jacobian of \texorpdfstring{\( A(\vec{h}) \)}{A(h)}}

As we have an expression for \( A(\vec{h}) \) in terms of \( \vec{h} \), we can compute the Jacobian of this vector function to make it easier for the root-finding algorithm to find a solution.

\subsubsection*{Reduction of degrees of freedom}
Recall that there are infinitely many \( \vec{h} \) satisfying \cref{eq:area constraint}. We may be tempted to leave \( \vec{h} \) arbitrary or, as it was done in \cite{BianchiDonaSpeziale2011}, to require \( h_i\geq 0 \), i.e., that the origin lies inside the polyhedron. Either way, however, the solution \( \vec{h}^* \) of \cref{eq:area constraint} is only defined up to 3 degrees of freedom, effectively resulting in infinitely many roots for the function
\begin{equation}
g: \vec{h} \in \R^F \mapsto A(\vec{h}) - A^0 \in \R^F.
\end{equation}
To see why this is a problem, let us take a closer look on how \( A(\vec{h}) \) depends on \( \vec{h} \). Reference \cite{BianchiDonaSpeziale2011} erroneously states that this function is \emph{quadratic} in \( \vec{h} \), being the sum of products of linear functions. The functions \( \lambda_{ij}(\vec{h}) \), however, are only \emph{piecewise linear}, as they are defined in terms of maxima and minima, just as a function of the form
\begin{equation}
f(x,y):= \max \set{x+y,2x-y} \equiv
\begin{cases}
x+y & \casesif x \leq 2y\\
2x-y & \casesif x > 2y.
\end{cases}
\end{equation}
An important consequence of this fact is that \( \lambda_{ij} \) is only piecewise differentiable, similarly to how, for the previous example,
\begin{equation}
\pder{f}{x}(x,y) = 
\begin{cases}
1 & \casesif x < 2y\\
2 & \casesif x > 2y
\end{cases}
\end{equation}
but it is undefined for all the points in \( \set{(2y,y) \setst y \in \R} \subset \R^2 \). It follows in particular that, because areas are necessarily non-negative, that \( g_i(\vec{h}) \) can become \emph{locally constant} when \( A_i(\vec{h}) = 0 \). For example, if
\begin{subequations}
\begin{alignat}{7}
&\vec{u}_1 = (1,0,0) &\qquad& \vec{u}_2 =(0,1,0) &\qquad& \vec{u}_3 =(0,0,1)\\
&\vec{u}_4 = (-1,0,0) && \vec{u}_5=(0,-1,0) && \vec{u}_6 =(0,0,-1)\\
&\vec{u}_7 = \mathrlap{(\sqrt{2}/{2},\sqrt{2}/{2},0),}
\end{alignat}
\end{subequations}
one can easily see that
\begin{equation}
h_7 \geq \frac{\sqrt{2}}{2}(h_1+h_2) \quad \Rightarrow \quad A_7(\vec{h})=0.
\end{equation}
This is particularly bad when trying to find a root of \( g(\vec{h}) \), since from the point of view of the algorithm
\begin{equation}
h_7 = 2(h_1 + h_2)
\quad\text{and}\quad
h_7 = 2000(h_1 + h_2)
\end{equation}
are not distinguishable, i.e., it has no way to know that \( h_7 \) should be \emph{decreased} to get closer to the solution.

Taking these things into account, we can see why having three degrees of freedom is a problem: at each iteration the root-finding algorithm will try to get closer to a different root; regardless of the initial guess we choose, we cannot guarantee that the algorithm is not going to be led to a situation were at least one of the \( g_i(\vec{h}) \) is constant and will fail to converge.

We will overcome this situation by fixing the degrees of freedom: instead of requiring that the origin is inside the polyhedron, we fix its position once and for all. A convenient location for the origin is at the intersection of three of the planes \( Q_i \). To ensure that the intersection is in fact a point, we choose three planes such that the respective unit normals span \( \R^3 \), which can always be done\footnote{Three independent unit vectors can efficiently be found in the following way: choose one vector \( \vec{u}_{i_1} \), then cycle through the others to find one, say \( \vec{u}_{i_2} \), with \( \abs{\braket{\vec{u}_{i_1},\vec{u}_{i_2}}} \neq 1 \). Finally, cycle through the remaining ones to find \( \vec{u}_{i_3} \) such that \( \braket{\vec{u}_{i_3},\vec{u}_{i_1} \times \vec{u}_{i_2}} \equiv a^{i_3}_{i_1 i_2}\neq 0 \). This method is efficient, as the dot products and the numbers \( a^k_{ij} \) have to be computed anyway to calculate \( A(\vec{h}) \).};
for simplicity we are going to assume that the \( \vec{u}_i \) are ordered such that
\begin{equation}
\Span\set{\vec{u}_1,\vec{u}_2,\vec{u}_3}=\R^3
\end{equation}
and choose the planes \( Q_1\), \( Q_2 \), \( Q_3 \). It follows that these planes are at distance \( 0 \) from the origin, i.e., we are only going to look at polyhedra of the form
\begin{equation}
P(0,0,0,\widetilde{\vec{h}}),\quad \widetilde{\vec{h}} \in \R^{F-3}.
\end{equation}
The function
\begin{equation}
\widetilde{g}:\widetilde{\vec{h}} \in \R^{F-3} \mapsto A(0,0,0,\widetilde{\vec{h}}) - A^0 \in \R^F
\end{equation}
has a \emph{unique root}, which we can now find with a root-finding algorithm.

\subsubsection*{Computing the Jacobian}
Having reduced the degrees of freedom, we can compute the jacobian of the function \( \widetilde{g} \). Let greek letters \(\alpha,\beta \) be indices ranging in \( \set{4,\dotsc,F} \). The Jacobian is the \( F\times(F-3) \) matrix with entries
\begin{equation}
J_{i\beta}(\widetilde{\vec{h}}) = \pder{\widetilde{g}_i}{h_\beta}(\widetilde{\vec{h}}) \equiv \pder{A_i}{h_\beta}(0,0,0,\widetilde{\vec{h}}).
\end{equation}
Recall that \( A(\vec{h}) \) is only piecewise differentiable, so the Jacobian can be discontinuous; when this happens we are going to assign a value for \( J \) at the discontinuity anyway by randomly choosing from one of the limiting values.
Since
\begin{equation}
A_i = \frac{1}{2}\sum_{j\neq i} \lambda_{ij} r_{ij},
\end{equation}
we have
\begin{equation}\label{eq:Jacobian entry}
\pder{A_i}{h_\beta} = \frac{1}{2}\sum_{j\neq i}\paren*{\pder{\lambda_{ij}}{h_\beta}r_{ij} + \lambda_{ij}\pder{r_{ij}}{h_\beta}};
\end{equation}
we can easily compute
\begin{equation}
\label{eq: derivative of r}
\pder{r_{ij}}{h_\beta} = \delta_{\beta j} - \delta_{\beta i} \braket{\vec{u}_i,\vec{u}_j},
\end{equation}
but to find the derivatives of \( \lambda_{ij} \) we first need to find those of \( \lambda^\text{max}_{ij} \) and \( \lambda^\text{min}_{ij} \), which depend on the \( b^k_{ij} \).
We can make explicit the dependence of \( b^k_{ij} \) on \( \widetilde{\vec{h}} \) by introducing the quantities
\begin{equation}
N^k_{ij} := \frac{\braket{\vec{u}_i,\vec{u}_j} \braket{\vec{u}_k,\vec{u}_j} - \braket{\vec{u}_k,\vec{u}_i}}{\norm{\vec{u}_i\times\vec{u}_j}^2},
\end{equation}
in terms of which we can write\footnote{Note that \( N^k_{ij} \) is not defined when \( \vec{u}_j = -\vec{u}_i \), but in this case \( \lambda_{ij} =0\). }
\begin{equation}
b^k_{ij} = h_k + h_i N^k_{ij} + h_j N^k_{ji}.
\end{equation}
Note that
\begin{equation}
\lambda^\text{max}_{ij} = \min_{k} \set[\bigg]{\frac{b^k_{ij}}{a^k_{ij}} \setst[\bigg] a^k_{ij}>0}
\equiv
\max_{k} \set[\bigg]{\frac{b^k_{ij}}{a^k_{ji}} \setst[\bigg] a^k_{ji}<0} = \lambda^\text{min}_{ji}.
\end{equation}
When \( \lambda_{ij} \neq 0 \), choose
\begin{equation}
k_\text{max}^{ij} \in \argmin_{k} \set[\bigg]{\frac{b^k_{ij}}{a^k_{ij}} \setst[\bigg] a^k_{ij}>0},
\quad
k_\text{min}^{ij} \in \argmax_{k} \set[\bigg]{\frac{b^k_{ij}}{a^k_{ij}} \setst[\bigg] a^k_{ij}<0};
\end{equation}
if the \( \argmin \) or \( \argmax \) contain more than one element, we can choose from them randomly, as long as we ensure
\begin{equation}
k_\text{max}^{ij} = k_\text{min}^{ji}.
\end{equation}
It follows that, when \( \lambda_{ij} \neq 0 \),
\begin{equation}
\lambda^\text{max}_{ij}=\frac{b^{k_\text{max}^{ij}}_{ij}}{a^{k_\text{max}^{ij}}_{ij}},
\quad
\lambda^\text{min}_{ij}=\frac{b^{k_\text{min}^{ij}}_{ij}}{a^{k_\text{min}^{ij}}_{ij}}.
\end{equation}
Introducing the numbers
\begin{align}
c_{ij} &:= \bracks{\lambda_{ij}\neq 0} \paren*{
\frac{N^{k_\text{max}^{ij}}_{ij}}{a^{k_\text{max}^{ij}}_{ij}}
-
\frac{N^{k_\text{min}^{ij}}_{ij}}{a^{k_\text{min}^{ij}}_{ij}}}
\equiv
\bracks{\lambda_{ij}\neq 0} \paren*{
\frac{N^{k_\text{max}^{ji}}_{ij}}{a^{k_\text{max}^{ji}}_{ji}}
-
\frac{N^{k_\text{min}^{ji}}_{ij}}{a^{k_\text{min}^{ji}}_{ji}}}
\\
y^k_{ij} &:= \bracks{\lambda_{ij}\neq 0}  \frac{1}{a^k_{ij}} \paren[\big]{\bracks{k=k^{ij}_\text{max}} - \bracks{k=k^{ij}_\text{min}}},
\end{align}
where we used the \emph{Iverson bracket}
\begin{equation}
\bracks{P}=
\begin{cases}
1 & \casesif \mbox{\( P \) is true}\\
0 & \casesif \mbox{\( P \) is false},
\end{cases}
\end{equation}
we see that
\begin{equation}\label{eq: derivative of lambda}
\pder{\lambda_{ij}}{h_\beta}
= \delta_{\beta i} c_{ij}
+ \delta_{\beta j} c_{ji}
+ y^\beta_{ij},
\end{equation}
which automatically includes the case \( \lambda_{ij} = 0 \). The Jacobian is finally computed by substituting \cref{eq: derivative of r,eq: derivative of lambda} into \cref{eq:Jacobian entry}.

Pseudocode for the computation of the Jacobian can be found in \cref{alg:jacobian}. The time complexity for the computation is \( O(F^3) \). In the Python implementation the Jacobian and the areas are calculated together, to avoid repetitions.

\begin{algorithm}[H]
\caption{Jacobian}
\label{alg:jacobian}
\begin{algorithmic}[1]
\Procedure{$\text{Jacobian}$}{$\vec{h}$}
\State{\( J \gets \vec{0}_{n,n-3} \)}\Comment{Initialise Jacobian (\( n\times(n-3) \) zero matrix)}
\State Compute \( k^{ij}_\text{max} \) and \( k^{ij}_\text{min} \):
\For{$i\in\set{1,\dotsc,F}$}
\For{$j\neq i$}
\State \( k^{ij}_\text{max} \gets \text{random element of } \argmin_{a^k_{ij}>0}\set*{b^k_{ij}(\vec{h})/a^k_{ij}} \)
\State \( k^{ij}_\text{min} \gets \text{random element of } \argmax_{a^k_{ij}<0}\set*{b^k_{ij}(\vec{h})/a^k_{ij}} \)
\EndFor
\EndFor
\State Compute each Jacobian entry \( J_{i\beta} \):
\For{$i\in\set{1,\dotsc,F}$, \( \beta\in\set{3,\dotsc,F} \)}
\For{$j \neq i$}
\If {\( \beta = i \)}
\State \( J_{i\beta} \gets J_{i\beta} + \frac{1}{2} \bracks*{c_{\beta j}(\vec{h}) r_{ij}(\vec{h}) - \braket{\vec{u}_i,\vec{u}_j} \lambda_{ij}(\vec{h})}  \)
\Comment Note that \( y^i_{ij} = y^j_{ij} = 0 \)
\ElsIf {\( \beta = j \)}
\State \( J_{i\beta} \gets J_{i\beta} + \frac{1}{2} \bracks*{c_{\beta i}(\vec{h}) r_{ij}(\vec{h}) + \lambda_{ij}(\vec{h})}  \)
\Else
\State \( J_{i\beta} \gets J_{i\beta} + \frac{1}{2}y^\beta_{ij}(\vec{h}) r_{ij}(\vec{h}) \)
\EndIf
\EndFor

\EndFor
\State \textbf{return} $J$
\EndProcedure
\end{algorithmic}
\end{algorithm}

\subsection{Finding the Polyhedron}
\begin{figure}
\centering
\foreach \x in {25,50,100}{
\subfloat[\( F=\x\)]{
\includegraphics[height=5.1cm,trim={0 -1cm 0 0},clip]{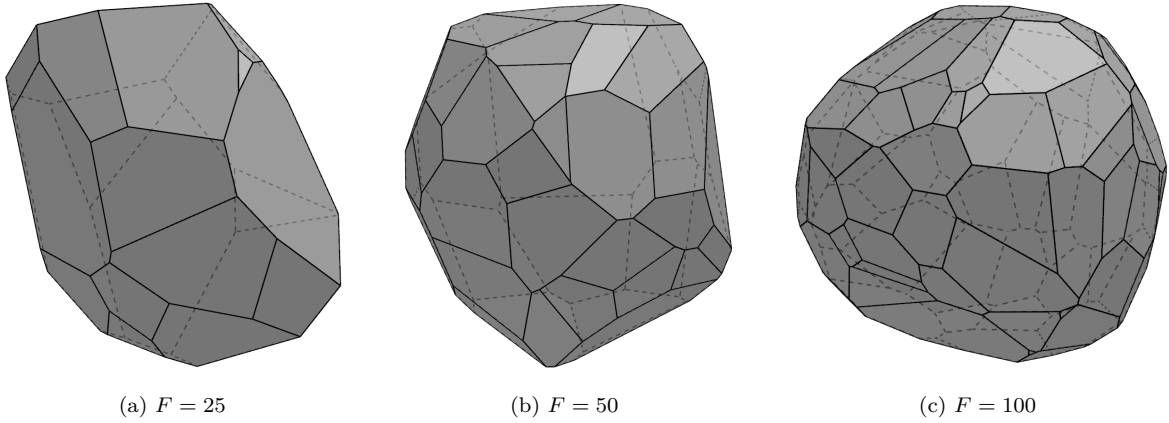}
}
\hfill
}
\caption{Examples of reconstructed polyhedra with number of faces \( F=25,50,100 \). In each case \( F-1 \) of the normals and areas were uniformly distributed respectively on \( \S^2 \) and the interval \( (0,1] \), while the last normal and area were constructed so that \( \sum_i A^0_i \vec{u}_i = 0 \). }
\label{fig:examples}
\end{figure}
We now have all the ingredients to reconstruct the polyhedron. The approach followed here greatly differs from the one presented in \cite{BianchiDonaSpeziale2011}. In the latter, the polyhedron is reconstructed by \emph{minimising} the function
\begin{equation}\label{eq:minimise}
\sum_{i} (A_i - A^0_i)^2.
\end{equation}
There are two problems with this approach.
\begin{itemize}
\item There is no guarantee that the function only has one minimum. Depending on the initial guess supplied to the algorithm, a local minimum for which the function doe not vanish may be found.
\item Since the algorithm is trying to minimise the function instead of finding a root, it can end up being drawn to a value of \( \vec{h} \) that makes the function constant, i.e., \(  A(\vec{h}) = \vec{0} \).
\end{itemize}
Both situations are dangerous, as the algorithm will fail to reconstruct the right polyhedron, but it will \emph{believe it succeeded}. The author has tried to reproduce the algorithm presented in \cite{BianchiDonaSpeziale2011}, but has found minimising \eqref{eq:minimise} to be unreliable, even when fixing the degrees of freedom.

The approach we are going to follow, as mentioned already, is to find the roots of the function \( \widetilde{g} \). The algorithm used in the Python implementation is the Levenberg--Marquardt as implemented in \emph{SciPy}~\cite{scipy}.
To avoid a situation in which some entries of the Jacobian vanish, it is important to choose the initial guess \( \widetilde{\vec{h}}_0 \) in such a way that
\begin{equation}
A_i(0,0,0,\widetilde{\vec{h}}_0) > 0,\quad \forall i.
\end{equation}
A reliable way to do so is to choose \( \widetilde{\vec{h}}_0 \) such that each \( Q_i \) is at a fixed distance \( D \) from a given point \( \vec{c} \) inside the polyhedron. There is a unique \( \vec{c} \) satisfying the requirements for \( D \) fixed: in fact, note that it must be\footnote{Recall that \( \Span\set{\vec{u}_1,\vec{u}_2,\vec{u}_3}=\R^3 \) and that the distance from a point \( \vec{x}\in P(\vec{h}) \) and \( Q_i \) is \( \sigma(\vec{x},Q_i)=h_i - \braket{\vec{u}_i,\vec{x}}\geq 0 \). }
\begin{equation}
\begin{cases}
\vec{c} = \alpha_1 \vec{u}_1 + \alpha_2 \vec{u}_2+ \alpha_3 \vec{u}_3
\\
\braket{\vec{u}_i,\vec{c}} = h_i - D,\quad i=1,\dotsc,F,
\end{cases}
\end{equation}
and in particular
\begin{equation}
\braket{\vec{u}_1,\vec{c}} = \braket{\vec{u}_2,\vec{c}} = \braket{\vec{u}_3,\vec{c}} =- D,
\end{equation}
that is
\begin{equation}\label{eq:matrix system}
\begin{pmatrix}
1 & \braket{\vec{u}_1,\vec{u}_2} & \braket{\vec{u}_1,\vec{u}_3}\\
\braket{\vec{u}_2,\vec{u}_1} & 1 & \braket{\vec{u}_2,\vec{u}_3}\\
\braket{\vec{u}_3,\vec{u}_1} & \braket{\vec{u}_3,\vec{u}_2} & 1
\end{pmatrix}
\begin{pmatrix}
\alpha_1 \\ \alpha_2 \\ \alpha_3
\end{pmatrix}
=
\begin{pmatrix}
-D \\ -D \\ -D
\end{pmatrix};
\end{equation}
there is a unique solution to \eqref{eq:matrix system}, since the \( 3\times 3 \) matrix on the l.h.s. is the \emph{Gram matrix} of the independent vectors \( \vec{u}_1 \), \( \vec{u}_2 \), \( \vec{u}_3 \), which implies it has non-zero determinant. Once \( \vec{c} \) is fixed, the initial guess is given by
\begin{equation}
h^0_\alpha = D + \braket{\vec{u}_\alpha,\vec{c}}.
\end{equation}
It only remains to fix \( D \). To do so, we first choose some value, say \( D_0 = 1 \), and use it to compute \( \widetilde{\vec{h}}_0 \). If we denote, with some abuse of notation,
\begin{equation}
A_i(D) := A_i(0,0,0,\widetilde{\vec{h}}_0(D)),
\end{equation}
one can easily see that
\begin{equation}
A_i(\lambda D) = \lambda^2 A_i(D), \quad \lambda>0,
\end{equation}
so that in particular
\begin{equation}
A_i(D) = D^2 A_i(D_0).
\end{equation}
We are going to choose \( D \) such that, in average,
\begin{equation}
\frac{A^0_i}{A_i(D)}= \frac{A^0_i}{D^2 A_i(D_0)} = 1,
\end{equation}
i.e.
\begin{equation}
D = \sqrt{\frac{1}{F}\sum_i \frac{A^0_i}{A_i(D_0)}};
\end{equation}
the updated \( \widetilde{\vec{h}}_0 \) is then given by
\begin{equation}
\widetilde{\vec{h}}_0 \rightarrow D \widetilde{\vec{h}}_0.
\end{equation}
Some examples of polyhedra reconstructed with the algorithm can be seen in \cref{fig:examples}

\section{Concluding remarks}\label{sec:conclusions}

To summarise, we have introduced an algorithm to reconstruct a convex polyhedron, given its face normals and face areas. To do so, a vector function measuring how distant the face areas of a polyhedron \( P(\vec{h}) \)---having the appropriate face normals--- are from the correct values was explicitly constructed; the reconstructed polyhedron is obtained by finding the root of this function. The highlight of the algorithm, as well as the main differences with the pre-existing literature~\cite{BianchiDonaSpeziale2011}, were:
\begin{itemize}
\item the explicit construction of the Jacobian matrix for the vector function;
\item allowing negative \( h_i \), i.e. not requiring that the origin lie inside the polyhedron;
\item fixing the origin to be at a specific point, effectively fixing three of the variables \( h_i \), in order to have a unique root.
\end{itemize}
The actual root-finding was done using a Levenberg--Marquardt, as it proved reliable in the tests performed by the author. The author leaves the possibility of developing an \emph{ad hoc} root-finding technique for further investigations.

\printbibliography
\end{document}